%% file: ControllingChaosFaster.tex
\documentclass[aip, cha, preprint, groupedaddress]{revtex4-1}

\usepackage{amsmath, amssymb, amsthm}
\usepackage[english]{babel}
\usepackage{xy}
\usepackage{graphicx}
\usepackage{verbatim}
\usepackage{enumerate}
\usepackage{url}
\usepackage{subfigure}
\usepackage[refpage]{nomencl}
\usepackage{float}
\usepackage[sort&compress]{natbib}
\usepackage{times}

\include{commands}
\newcommand{\imagescaling}{0.9}
\newcommand{\CO}{(color online) }

\bibpunct{[}{]}{,}{s}{}{;}

\begin{document}


\title{Controlling Chaos Faster}

\author{Christian Bick${}^{a,b,c}$}
\author{Christoph Kolodziejski${}^{a,d}$}
\author{Marc Timme${}^{a,e}$}
\address{${}^{a}$Network Dynamics, Max Planck Institute for Dynamics and Self-Organization (MPIDS), 37077 G\"ottingen, Germany\\
${}^b$Bernstein Center for Computational Neuroscience (BCCN), 37077 G\"ottingen, Germany\\
${}^c$Institute for Mathematics, Georg--August--Universit\"at G\"ottingen, 37073 G\"ottingen, Germany\\
${}^d$III.~Physical Institute---Biophysics, Georg--August--Universit\"at G\"ottingen, 37077~G\"ottingen, Germany\\
${}^e$Institute for Nonlinear Dynamics, Georg--August--Universit\"at G\"ottingen, 37077~G\"ottingen, Germany}

\altaffiliation{CB currently at Department of Mathematics, Rice University, MS--136, 7500 Main St., Houston, TX 77005, USA}

\date{\today}

\begin{abstract}
Predictive Feedback Control is an easy-to-implement method to stabilize
unknown unstable periodic orbits in chaotic dynamical systems. Predictive 
Feedback Control is severely  limited because asymptotic convergence speed 
decreases with stronger instabilities which in turn are typical for larger 
target periods, rendering it harder to effectively stabilize
periodic orbits of large period. Here, we study stalled chaos control,
where the application of control is stalled to make use of the chaotic, 
uncontrolled dynamics, and introduce an adaptation paradigm to overcome 
this limitation and speed up convergence. 
This modified control scheme is not only capable of stabilizing more 
periodic orbits than the original Predictive Feedback Control but also 
speeds up convergence for 
typical chaotic maps, as illustrated in both theory and application. The 
proposed adaptation scheme provides a way to tune parameters online, 
yielding a broadly applicable, fast chaos control that converges reliably, 
even for periodic orbits of large period.
\end{abstract}


\maketitle


\begin{quotation}
Chaos control underlies a broad range of applications across physics 
and beyond. To successfully use chaos control schemes in applications,
different robustness and convergence properties need to be considered 
from a practical point of view. For instance, for control to be 
useful in praxis, a method does not only need to guarantee 
convergence to the desired state, but convergence also has to be 
sufficiently fast.

Predictive Feedback Control provides an easy-to-implement way to 
realize chaos control in discrete time dynamical systems (iterated 
maps). However, periodic orbits of larger periods are typically 
highly unstable, leading to slow convergence. Here, we systematically 
investigate a recently introduced extension of Predictive Feedback 
Control obtained by stalling control and complement it with an 
adaptation mechanism. The stalling of control, i.e., repeated 
transient interruption of control, takes advantage of the 
uncontrolled chaotic dynamics, thereby speeding up convergence. 
Adaptation provides a way to tune the control parameters online to 
values which yield optimal speed.

Specifically, we show how the efficiency of stalling control depends 
on both the local stability properties of the periodic orbits to be 
stabilized and the choice of control parameters. Furthermore, we derive 
conditions for stabilizability of periodic orbits in systems of higher 
dimensions. In addition to speeding up convergence, the gradient 
adaptation scheme presented also further increases the overall 
convergence reliability. Hence, Adaptive Stalled Predictive Feedback 
Control yields an easy-to-implement, noninvasive, fast, and reliable 
chaos control method for a broad scope of applications.
\end{quotation}

\section{Introduction}

Typically, chaotic attractors contain infinitely many unstable periodic 
orbits~\cite{Katok1995}. The goal of chaos control is to render
these orbits stable. After first being introduced in the seminal work
by Ott, Grebogi, and Yorke \cite{Ott1990} about two decades ago, it has
not only been hypothesized to be a mechanism exploited in biological
neural networks \cite{Rabinovich1998} but it has found its way into many
applications \cite{Scholl2007, Garfinkel1992} including
chaotic lasers, stabilization of cardiac rhythms, and more recently into
the control of autonomous robots \cite{Steingrube2010}.

Predictive Feedback Control (PFC) \cite{DeSousaVieira1996, Polyak2005}
is well suited for applications: little to no prior knowledge about the
system is required, it is non-invasive, i.e., control strength vanishes
upon convergence, and it is very easy to implement due to the nature
of the control transformation. In PFC, a prediction of the future
state of the system together with the current state is fed back
into the system as a control signal, similar to time-delayed feedback
control \cite{Pyragas1992}. In fact, it can be viewed as a special case
of a recent effort to determine all unstable periodic points of a discrete
time dynamical system \cite{Schmelcher1997, Schmelcher1998} which has
been studied and extended \cite{Pingel2000, Crofts2006, Doyon2002,
Davidchack1999} for its original purpose.

In any real world application not only the existence of parameters that
lead to stabilization, but also the speed of convergence is of importance.
Speed is crucial, for example, if a robot is controlled by stabilizing
periodic orbits in a chaotic attractor~\cite{Steingrube2010}, since 
the time
it needs to react to a changing environment is bounded by the time
the system needs to converge to a periodic orbit of a given period. In
most of the literature, however, speed of convergence has been
overlooked. Stabilizing periodic orbits of higher periods
becomes quite a challenge; due to the increasing instability of the
orbits, the PFC method yields only poor performance in terms of
asymptotic convergence
speed even when the control parameter is chosen optimally. Any method
optimizing speed within the PFC framework~\cite{Bick2010b} therefore
is subject to the same limitation.

In this article we investigate Stalled Predictive Feedback Control
(SPFC), a recently proposed extension of Predictive Feedback
Control that can overcome this ``speed limit''\cite{Bick2012}.
Here, we derive conditions for the local stability properties of
periodic orbits that imply stabilizability. Furthermore, we propose
an adaptation mechanism that is capable of tuning the control 
parameter online to reach optimal asymptotic convergence speed
within the regime of convergence. The resulting adaptive SPFC
method
is an easy-to-implement, non-invasive, and broadly applicable
chaos control method that stabilizes even periodic orbits of large
periods reliably without the need to fine-tune parameter values
a priori.

This article is organized as follows. In the following section,
we formally introduce the PFC method, briefly discuss its 
limitations and present SPFC as an alternative. The third
section is dedicated to an in-depth look at the SPFC method;
we identify regimes in parameter space in which stabilization
is successful. In the fourth section, we apply our algorithm
to ``typical'' maps with chaotic dynamics and calculate and 
compare convergence speeds.
Adaptive methods for the control parameter are explored in
Section~\ref{sec:Adaptation} before giving some concluding remarks.

\section{Preliminaries}\label{sec:Prelim}

Suppose $f: \Rn\to\Rn$ is a differentiable map such that the
iteration given by the evolution equation $x_{k+1}=f(x_k)$ gives
rise to a chaotic attractor $A\subset\Rn$ with a dense set of
unstable periodic orbits. We refer to such a map as a 
\emph{chaotic map}. Let $\FP(f)=\set{\xf\in\Rn}{f(\xf)=\xf}$ denote
the set of fixed points of $f$ and $\id$ the identity map
on $\Rn$. The main result of 
Schmelcher and Diakonos\cite{Schmelcher1998} reads as follows.

\begin{prop}\label{prop:PFC}
Suppose $\FPs(f)\subset\FP(f)$ is the set of fixed points such
that both $\drv{f}{\xf}$ and $\drv{f}{\xf}-\id$ are nonsingular
and diagonalizable (over $\C$). Then there exist finitely many
orthogonal matrices $M_k\in O(\maxdim)$, $k=1, \ldots, K$, such
that we have
\[\FPs(f) = \bigcup_{k=1}^{K}\Cc(f, M_k)\]
where the sets $\Cc(f, M_k)$ are characterized by the
the property that for $\xf\in\Cc(f, M_k)$ there exists 
$\mu\in(0,1)$ such that $\xf$ is a stable fixed point of the
map $g_{\mu, 1}$ obtained by the transformation
$S(\mu, M_k): f\mapsto\id+\mu M_k(f-\id)=g_{\mu, 1}.$
\end{prop}

\subsection{Predictive Feedback Control}

This result may be cast into a control method. Let~$\N$ denote the 
set of natural numbers. A periodic orbit of period $p\in\N$ is a fixed
point of the $p$th iterate of~$f$ denoted by
\[f_p:=\ite{f}{p}=\underbrace{f\circ\cdots\circ f}_{p \text{ times}},\]
and therefore we use the terms fixed point and periodic orbit
interchangeably depending on what is convenient in the context.
Let $\Per(f) = \bigcup_{p\in\N}\FP(f_p)$ denote the set of all periodic
points of~$f$.
Define the set of periodic orbits of minimal period~$p$ as
$\FP(f, p) = \set{\xf\in\FP(f_p)}{\ite{f}{q}(\xf)\neq\xf \text{ for }q<p}$.
Furthermore, we define $\FPs(f, p) = \FP(f, p)\cap \FPs(f_p)$.
Predictive Feedback Control is now a consequence of 
Proposition~\ref{prop:PFC} by replacing~$f$ with~$f_p$.

\begin{cor}
Let $p\in\N$. For every $\xf\in\FPs_g(f, p) := \FPs(f, p)\cap\left(\Cc(f_p, \id)
\cup\Cc(f_p, -\id)\right)$ there exists a $\mu\in(-1, 1)$ such that
$\xf$ is a stable fixed point of the \emph{Predictive Feedback
Control} method given by the iteration
\[x_{k+1} = g_{\mu, p}(x_k+1) := f_p(x_k)+\eta(x_k-f_p(x_k))\]
with $\eta = 1-\mu$ and \emph{control perturbation} 
$c_{\mu, p}(x) = \eta\left(x_k-f_p(x_k)\right)$.
\end{cor}

The elements of $\FPs_g(f, p)$ are referred to as \emph{PFC-stabilizable}
periodic orbits of period~$p$. The cardinality of the set $\FPs_g(f, p)$
depends on the chaotic map~$f$ and contains roughly half of the periodic
orbits of a given period in two-dimensional 
systems~\cite{Schmelcher1998, Pingel2000}.

Fix $\xf\in\FPs(f, p)$. Local stability of $g_{\mu, p}$ at $\xf$ is 
readily computed. Let $\drv{f}{x}$ denote the total derivative 
of~$f$ at~$x$ and suppose that $\lambda_j$, $j=1, \ldots, \maxdim$.
are the eigenvalues of the linearization. The derivative of 
$g_{\mu, p}$ at~$\xf$ evaluates to
$\drv{g_{\mu, p}}{x} = \id+\mu(\drv{f_p}{x}-\id).$
Hence, stability is determined by the eigenvalues of 
$\drv{g_{\mu, p}}{\xf}$ given by
\begin{equation}\label{eq:StabilityPFC}
\kappa_j(\mu) = 1+\mu(\lambda_j-1)
\end{equation}
for $j=1, \ldots, \maxdim$. Hence, $\xf\in\FPs_g(f, p)$ iff there
exists a $\mu_0\in(-1,1)$ such that the spectral radius 
$\vrho(\drv{g_{\mu_0, p}}{\xf})=\max_{j=1, \ldots, N}\abs{\kappa_j(\mu_0)}$
is smaller than one. In particular, for a two-dimensional system 
these are the periodic orbits of saddle type\cite{Pingel2000} with
stable direction $\lambda_1\in(-1, 1)$ and $\lambda_2<-1$. Note
that optimal convergence speed is achieved for the value of~$\mu$
which corresponds to the minimal spectral radius.

\subsection{Speed Limit of Predictive Feedback Control}

For increasing instability, however, the optimal convergence speed
becomes increasingly slow \cite{Bick2012}. This applies in particular
to periodic orbits of larger periods as the periodic orbits become 
increasingly unstable on average\cite{Cvitanovi2010} and 
asymptotic convergence speed decreases. Let $\card$ denote the 
cardinality of a set. The slowdown of PFC can be
explicitly calculated by evaluating the functions
{\allowdisplaybreaks
\begin{subequations}\label{eq:StabPerG}
\begin{align}
\label{eq:StabPerGFirst}
\underline{\rho}_g(p) &= 1-\min_{\xf\in\FPs_g(f,p)}\vrho_{\text{min}}^{g}(\xf),\\
\rho_g(p) &= 1-\frac{1}{\card(\FPs_g(f,p))}\sum_{\xf\in \FPs_g(f,p)}\vrho_{\text{min}}^{g}(\xf),\\
\label{eq:StabPerGLast}
\overline{\rho}_g(p) &= 1-\max_{\xf\in\FPs_g(f,p)}\vrho_{\text{min}}^{g}(\xf),
\end{align}
\end{subequations}
}that quantify to the best, average, and worst asymptotic convergence
speed for all periodic orbits of a given period respectively. 

The slowdown becomes explicit in specific examples. We 
evaluated these functions for a map which describes the evolution of 
a two-dimensional neuromodule\cite{Pasemann2002}.
Let $l_{11}=-22, l_{12}= 5.9, l_{21}=-6.6$, and $l_{22}= 0$
and define the sigmoidal function $\s(x)=(1+\exp(-x))^{-1}$. The 
dynamics of the neuromodule is given by the map $f:\Rr^2\to\Rr^2$ 
where
\begin{multline}\label{eq:Example2D}
f(x_1, x_2) = (l_{11}\s(x_1) + l_{12}\s(x_2) - 3.4,\\ l_{21}\s(x_1) + l_{22}\s(x_2) +3.8).
\end{multline}
The values of the functions~\eqref{eq:StabPerG} are depicted in
Figure~\ref{fig:SrPFC}. One can clearly see that even the lower bound
on asymptotic convergence speed for the PFC method, corresponding to
the smallest spectral radius as determined by $1-\underline{\rho}_g$,
approaches one exponentially on average for increasing periods. This
scaling of convergence speed of PFC is quite typical;
other maps with chaotic attractors, such as the H\'enon map 
exhibit a similar behavior when subject to PFC~\cite{Bick2012}.

\begin{figure} 
\includegraphics[scale=\imagescaling]{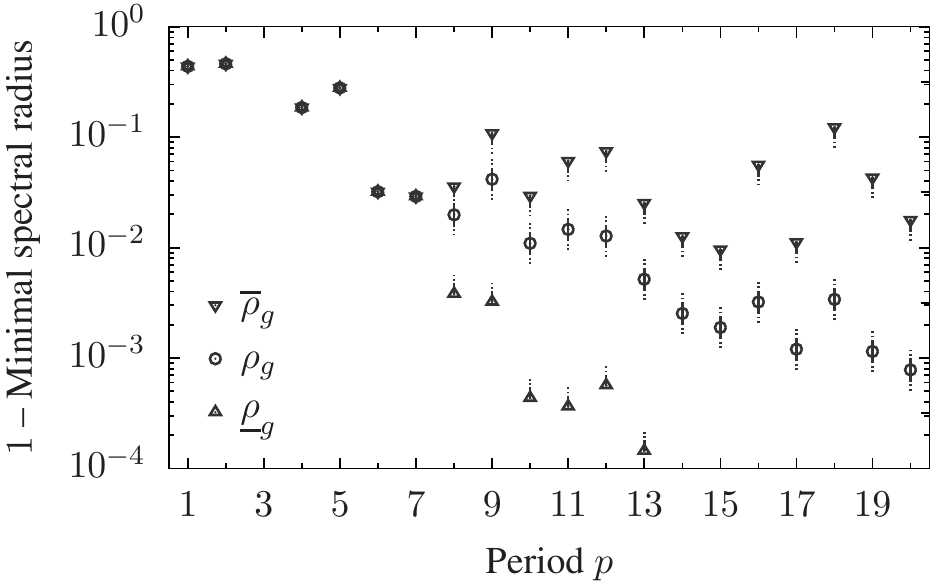}
\caption{\label{fig:SrPFC}Best, average, and worst asymptotic convergence
speed decreases as the period of periodic orbits of increases. Here, 
bounds on the spectral radius are plotted for the two-dimensional 
map~\eqref{eq:Example2D}.}
\end{figure}

\subsection{Stalled Predictive Feedback Chaos Control}

By making use of the uncontrolled dynamics, i.e., ``stalling control'',
it was recently shown that this speed limit may be overcome\cite{Bick2012}.
Stalled Predictive Feedback Control scheme is an extension of standard
Predictive Feedback Control. For a map~$\psi:\Rn\to\Rn$ define the
``zeroth iterate'' by $\ite{\psi}{0}:=\id$.

\begin{defn}\label{def:SPFC}
Suppose that the iteration of $F:\Rn\to\Rn$ defines a dynamical system.
For $M_k\in\sset{\pm\id}$ and $\mu\in\R$ let 
$S(\mu, M_k)(F)=\id+\mu M_k(F-\id)=:G_\mu$ denote the map obtained by
applying the Predictive Feedback Control transformation;
cf.~Proposition~\ref{prop:PFC}. For parameters $\prm,\prn\in\N_0=\N\cup\{0\}$
and $\mu\in\R$, the iteration of 
\begin{equation}\label{eq:DelayedStabDef}
H_{\mu}^{(\prm,\prn)} = \ite{\left(F\right)}{\prn}\circ\ite{\left(G_\mu\right)}{\prm}
\end{equation}
is referred to as Stalled Predictive Feedback Control.
\end{defn}

The function $H_{\mu}^{(\prm,\prn)}$ defined above stalls Predictive Feedback Control
in the following sense. In the PFC method, the control signal is applied
at every point in time. By iterating $H_{\mu}^{(\prm,\prn)}$ we ``stall'' the 
application of the control perturbation by adding extra evaluations of
the original, uncontrolled map~$F$.

Henceforth, we adopt the period-dependent notation introduced 
above: the uncontrolled dynamics were given by iterating 
$f:\Rn\to\Rn$ and the PFC transformed map is denoted 
by~$g_{\mu,p}$. Stalled Predictive Feedback Control is given
by the iteration of
\begin{equation}\label{eq:DelayedStab}
h_{\mu, p} = h_{\mu, p}^{(\prm,\prn)} := \ite{\left(f_p\right)}{\prn}\circ\ite{\left(g_{\mu, p}\right)}{\prm},
\end{equation}
where $\prm,\prn\in\N_0$ are parameters. By definition, we have
$h_{\mu, p}^{(0,1)} = f_{p}$ and we recover the original PFC method
for $h_{\mu, p}^{(1,0)} = g_{\mu, p}$. In general, we will omit the
superscript $(\prm,\prn)$ unless the choice is important.

\section{Stability of Stalled Predictive Feedback Chaos Control}
\label{sect:StallCtrlProp}

The stability of a periodic orbit in the controlled system depends 
on its stability properties for the uncontrolled dynamics. In this
section we derive criteria for a periodic orbit to be stabilizable
for Stalled Predictive Feedback Control.

\subsection{Local stability of periodic orbits for $h_{\mu, p}$}
\label{sec:LocStability}

The local stability properties of $h_{\mu, p}$ can be calculated
from $f_p$ and $g_{\mu, p}$.
By definition we have $\FP(f_p)\subset\FP(h_{\mu, p})$. Suppose
that
$\xf\in\FPs(f, p)$ and the eigenvalues of~$\drv{f_p}{\xf}$ are given
by $\lambda_j$ where $j=1, \dotsc, \maxdim$. Note that the eigenvectors
of~$\drv{g_{\mu, p}}{\xf}$ and~$\drv{f_p}{\xf}$ are the same.
Hence, the local stability properties of $h_{\mu, p}$ are
readily computed from the $\lambda_j$ and the local stability
properties of the PFC transformed map $g_{\mu, p}$ as given 
by~\eqref{eq:StabilityPFC}. The eigenvalues of the Jacobian of
$h_{\mu, p}$ at $\xf$ evaluate to
\[\Lambda_j = \lambda_j^{\prn}\kappa_j(\mu)^{\prm}=
\lambda_j^{\prn}\left(1+\mu(\lambda_j-1)\right)^{\prm}\]
for $j=1, \dotsc, \maxdim$. Hence, local stability at $\xf$ is
given by the spectral radius
\[\sr(\drv{h_{\mu, p}}{\xf}) = \max_{j=1, \dotsc, \maxdim}\abs{\Lambda_j}.\]
If all eigenvalues are of modulus smaller than one, the fixed 
point~$\xf$ is stable for~$h_{\mu, p}$. 
In other words, a periodic orbit $\xf\in\FPs(f, p)$ is called
\emph{SPFC-stabilizable} if there are parameters $\prm,\prn\in\N_0$
and $\mu\in(-1,1)$ such that
\[\sr\big(\drv{h_{\mu, p}^{(\prm,\prn)}}{\xf}\big)<1.\]

Let $\FPs_h(f, p)$ denote the set of SPFC-stabilizable periodic
orbits and, clearly, $\FPs_g(f, p)\subset\FPs_h(f, p),$ that is,
every PFC-stabilizable periodic orbit is also SPFC-stabilizable.

To compare the ``performance'' of Stalled Predictive Feedback
Control with that of original Predictive Feedback Control we have
to rescale the stability properties. Since~$h_{\mu, p}^{(m,n)}$
contains $n+m$ evaluations of~$f_p$ we take the $(m+n)$th root
to obtain functions
\[\alts{l}_j(\prm,\prn,\mu) = \abs{\lambda_j^{\prn}\left(1+\mu(\lambda_j-1)\right)^{\prm}}^{\frac{1}{\prm+\prn}},\]
where $j=1, \dotsc, \maxdim$. With the parameter
$\alpha=\frac{\prn}{\prm+\prn}$
we thus obtain an equivalent set of functions
\begin{equation}\label{eq:SingStab}
l_j(\alpha, \mu) = \abs{\lambda_j}^\alpha\abs{\left(1+\mu(\lambda_j-1)\right)}^{1-\alpha}
\end{equation}
for $j=1, \dotsc, \maxdim$ which determine the local stability
properties of $h_{\mu, p}$ rescaled to a single evaluation of~$f_p$.
Conversely, for any rational
$\alpha\in [0, 1]\cap\Q$
we obtain a pair $(\prm,\prn)$. In the following, we refer to 
both~$\alpha$ and the pair $\prm,\prn$ as \emph{stalling parameters}, 
depending what is convenient in the context. When using the stalling 
parameter~$\alpha$, we may also write $h_{\mu, p}^\alpha$.

Rescaled local stability of Stalled Predictive Feedback Control
for a given periodic orbit $\xf\in\FPs(f,p)$ of period $p$ is hence
determined by the \emph{stability function}
\begin{equation}\label{eq:StabilityFunction}
\vrho_\xf(\alpha, \mu) = \max_{j=1,\dotsc, \maxdim}l_j(\alpha, \mu).
\end{equation}
In comparison to the original Predictive Feedback Control, Stalled
Predictive Feedback Control depends on two parameters: the control
parameter~$\mu$ and the stalling parameter~$\alpha$.

\subsection{Conditions for stabilizability}
\label{sec:Stabilizability}

To derive conditions for SPFC-stabilizability, consider some general
properties of functions of type~\eqref{eq:SingStab}.
Fix $w\in\Cs := \pn{\C}{0}$. Let $\So:= \set{z\in\C}{\abs{z}=1} \cong 
\Rr/2\pi\Z$ denote the unit
circle. We will choose a realization to describe elements of $\So$
depending
on what is convenient in the context. Consider the function 
$L_w:\R^2\to\R$ given by
\[L_w(\alpha, \mu) := \abs{w}^\alpha\abs{1+\mu(w-1)}^{1-\alpha}.\]
By definition, we have $L_w(0, 0) = 1$ and in a sufficiently small
open ball $V$ around $(0, 0)$ the function $L_w$ is differentiable
and the derivative is bounded away from zero. Hence, in this
ball the curve defined by
\[V_0 := \set{(\alpha, \mu)\in V}{L_w(\alpha, \mu)=1}\]
is a one-dimensional submanifold of $\Rr^2$. If $V$ is chosen small
enough, it may be written as a disjoint union 
\[V = V_0 \cup V_+ \cup V_-\]
where $V_+ = \set{(\alpha, \mu)\in V}{L_w(\alpha, \mu) > 1}$ and
$V_- = \set{(\alpha, \mu)\in V}{L_w(\alpha, \mu) < 1}$.

The goal is to get a linearized description close to the origin. 
Let~$\grad$ denote the gradient and $\langle\,\cdot\,,\cdot\,\rangle$
the usual Euclidean scalar product. Define the line
\begin{equation}\label{eq:Halfplane}
\gamma(w) = \set{x\in\Rr^2}{\langle\grad(L_w)|_{(0, 0)}, x\rangle = 0}
\end{equation}
which is tangent to $V_0$ at the origin. Let 
\[\Hp:=\tset{x\in\Rr^2}{\langle\grad(L_w)|_{(0, 0)}, x\rangle < 0}\]
denote
one of the half planes defined by the line~$\gamma(w)$. 
Moreover, the sets $Q_j := \big(\frac{(j-1)\pi}{2}, \frac{j\pi}{2}\big)$
for $j\in\sset{1, 2, 3, 4}$
denote the open segments of $\So$ that lie in one of the four quadrants
of~$\Rr^2$.

\begin{defn}\label{defn:StabTuple}
Suppose that $w\in\Cs$. The connected subset $C_w := \Hp\cap\So$ is
called the domain of stability of~$w$. For a tuple 
$\tilde w=(w_1, \dotsc, w_\maxdim)\in(\Cs)^\maxdim$ define the
domain of stability to be
\begin{equation}\label{eq:SomOfStab}
C_{\tilde w} := \bigcap_{k=1}^{\maxdim}C_{w_k}.
\end{equation}
If $C_{\tilde w}\cap
\overline{\left(Q_1\cup Q_4\right)}\neq\emptyset$ then the tuple 
$\tilde w$ is called {stabilizable}.
\end{defn}

In a sufficiently small neighborhood $U\subset V$ of the origin, the 
``linearized'' version of~$V_-$ is given by the set $\Hp\cap U$.

\begin{lem}\label{lem:TupleStab}
If the domain of stablity $C_w$ of a tuple 
$w=(w_1, \dotsc, w_\maxdim)\in(\Cs)^\maxdim$ is nonempty then there
exist $(\mu_0, \alpha_0)$ such that $L_{w_j}(\mu_0, \alpha_0) < 1$
for all $j=1,\dotsc,\maxdim$. If the tuple $w$ is stabilizable then
then $\alpha_0$ may be chosen such that $\alpha_0\geq 0$.
\end{lem}

\begin{proof}
Suppose that $V_-$ and $\Hp$ are defined as above. Because of
continuity, for every $w\in\Cs$ there exists an open ball 
$B_w\subset V_-\cap\Hp$ that is tangent to the origin. If a
tuple $\tilde w = (w_1, \dotsc, w_\maxdim)$ has nonempty domain of
stability $C_{\tilde w}$ then 
\[B:=\bigcap_{j=1}^\maxdim B_{w_j}\neq \emptyset.\]
By construction, any $(\mu_0, \alpha_0)\in B$ has the desired property.

If in addition $w$ is stabilizable then the intersection 
$B\cap \set{(x,y)\in\R^2}{x\geq 0}$ is not empty. This proves
the second assertion.
\end{proof}

The domain of stability is determined by the gradient
of $L_w$ at the origin. Let $\ln$ denote the (real) natural logarithm.
We have $\grad (L_w)|_{(0, 0)} = (\ln{\abs{w}}, \Re(w)-1)$. Define 
\begin{align*}
R_1 &:= \set{z\in \C}{\Re(z) > 1},\\
R_2 &:=\set{z\in \C}{\abs{z} < 1},\\
R_3&:=\set{z\in \C}{\abs{z} > 1\text{, }\Re(z) < 1}.
\end{align*}
These regions are
sketched in Figure~\ref{fig:Stabilizability}(a). If $w\in R_1$ then
$\norm{\grad(L_w)|_{(0, 0)}}^{-1}\cdot\grad(L_w)|_{(0, 0)}\in Q_1$ and
therefore $Q_3\subset C_w$. Similarly, if $w\in R_2$ then 
$Q_1\subset C_w$ and if $w\in R_3$ then $Q_2\subset C_w$ 
(Figure~\ref{fig:Stabilizability}(b)--(d)). For $w$ on the boundary of
the $R_k$ the gradient lies on one of the coordinate axes and we obtain
similar conditions.

\begin{figure*}[t] 
\subfigure[\ Complex plane]{\includegraphics[scale=\imagescaling]{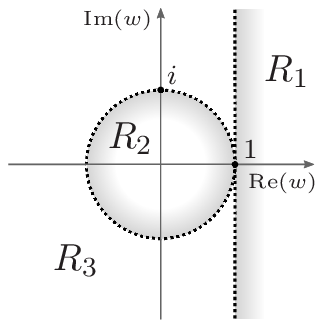}}\qquad\qquad
\subfigure[\ $w\in R_1$]{\includegraphics[scale=\imagescaling]{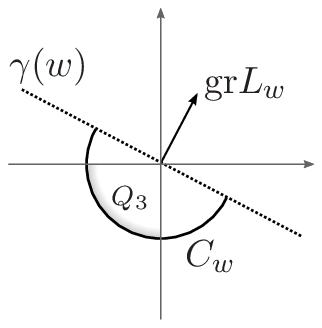}}\qquad
\subfigure[\ $w\in R_2$]{\includegraphics[scale=\imagescaling]{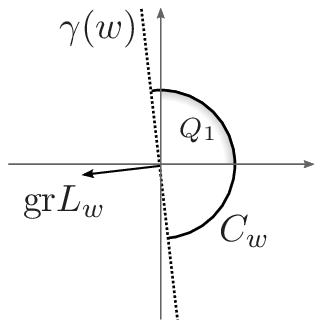}}\qquad
\subfigure[\ $w\in R_3$]{\includegraphics[scale=\imagescaling]{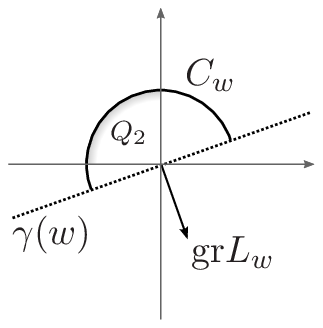}}
\caption{\label{fig:Stabilizability} Stabilizability regions for $w\in\C$ are
shown in Panel~(a) and the corresponding domains of stability 
$C_w$, as given by~\eqref{eq:SomOfStab}, for the three cases in Panels~(b)--(d). Here,
$\text{gr}L_w = \grad (L_w)|_{(0, 0)}$.}
\end{figure*}

These observations have implications for stabilizability for a
tuple $(w_1, \dotsc, w_\maxdim)$: if for any fixed $k\in\sset{1,2,3}$
all $w_j\in R_k$ for $j=1, \dotsc, \maxdim$ then the tuple is
stabilizable. Furthermore, if either $w_j\in R_1 \cup R_{3}$ or
$w_j\in R_2 \cup R_{3}$ for all $j=1,\dotsc,\maxdim$ then the tuple
is stabilizable. For any other combination the condition of
stabilizability is more difficult; in two dimensions linear
dependence of the
gradients tells us for $(w_1, w_2)$ with $w_1 \in R_1$ and
$w_2 \in R_2$ the tuple is stabilizable iff
\begin{equation}\label{eq:StabCond}
\ln(\abs{w_2})\Re(w_1) \neq \ln(\abs{w_1})\Re(w_2).
\end{equation}
Note that this condition is satisfied for a set of full Lebesgue
measure. 

\begin{rem}\label{rem:Conjugate}
Note that stabilizability is not affected by taking the complex
conjugate. Hence, stabilizability of a tuple of nonzero complex 
numbers as defined in Definition~\ref{defn:StabTuple} does not 
change when an entry of the tuple is replaced by its complex 
conjugate.
\end{rem}

For $w=0$ the function $L_0$ has a discontinuity at $\alpha = 0$.
In case $\alpha>0$ we have $L_0(\alpha, \mu)=0$ and for $\alpha=0$
and $\mu\in(-1,1)$ we have $L_0(0, \mu) = 1-\mu$. Therefore,
define the domain of stabilizability of zero to be 
$C_0=\{0, \frac{\pi}{2}\}\cup Q_1\cup Q_4$. For~\mbox{$\alpha>0$}, stabilizability
of a tuple with one component equal to zero may be reduced to
stabilizability of the ``reduced'' tuple where the zero entry is
omitted.

With the notation as above, we are now able to relate these general results
to the local stability properties of a
given periodic orbit.

\begin{defn} Suppose that $\xf\in\FPs(f, p)$ is a periodic orbit 
of~$f$ and suppose that the eigenvalues of $\drv{f_p}{\xf}$ are
given by~$\lambda_j$ with
$j=1,\dotsc,\maxdim$. The periodic orbit is called {locally
stabilizable} if the tuple $\lambda = (\lambda_1, \dotsc, 
\lambda_{\maxdim})$ is stabilizable as a tuple, as defined in 
Definition~\ref{defn:StabTuple}.
\end{defn}

This definition links the notion of stabilizability of a tuple
defined above and the local dynamics close to a periodic orbit.
Recall the notion of uniform hyperbolicity~\cite{Katok1995}.
Suppose that a differentiable function~$f$ defines a discrete time
dynamical system on~$\Rn$. We call an $f$-invariant set $A\subset\Rn$
hyperbolic if for every $x\in A$ no eigenvalue of~$\drv{f}{x}$ is
of absolute value one.

\begin{prop}
Suppose that the chaotic map $f:\Rn\to\Rn$ gives rise to a hyperbolic
attractor and for $\xf\in\FPs(f, p)$ let $\lambda =
(\lambda_1, \dotsc, \lambda_\maxdim)$ denote the eigenvalues of
$\drv{f_p}{\xf}$. If~$\xf$ is locally stabilizable then~$\xf$ is
SPFC-stabilizable. Moreover, if the domain of stability $C_\lambda$
satisfies
\[\left\{\frac{\pi}{2}, \frac{3\pi}{2}\right\}\cap C_\lambda\neq\emptyset\]
then $\xf$ is PFC-stabilizable.
\end{prop}

\begin{proof}
If a periodic orbit $\xf$ is locally stabilizable, then tuple~$\lambda$
is stabilizable. Thus,
according to Lemma~\ref{lem:TupleStab}, there are parameters 
$(\alpha_0, \mu_0)$ such that $L_{\lambda_j}(\alpha_0, \mu_0) < 1$
for all $j=1,\dotsc,\maxdim$ simultaneously. Recall that local
stability of  $h_{\mu_0, p}^{\alpha_0}$ at $\xf$ is given by
$l_j(\alpha, \mu)=L_{\lambda_j}(\alpha, \mu)$ according to
Equation~\eqref{eq:SingStab}. Therefore, local stability of
a periodic orbit is equivalent to the existence of parameters
$(\alpha_0, \mu_0)$ with $\alpha_0\geq 0$ and
\[\sr(\drv{h_{\mu_0, p}^{\alpha_0}}{\xf})<1.\]
which proves the first statement.

If $\tsset{\frac{\pi}{2}, \frac{3\pi}{2}}\cap C_\lambda\neq\emptyset$
then there exists a parameter $\mu_0$ such that
$\sr(\drv{h_{\mu_0, p}^{0}}{\xf})<1$. Since
Stalled Predictive Feedback Control reduces to classical Predictive
Feedback Control for a stalling parameter of $\alpha=0$, the claim
follows.
\end{proof}

The conditions derived for stabilizability of tuples translate
directly into conditions on the local stability properties of
a periodic orbit. For dynamics in two dimensions we obtain the
following immediate consequence.

\begin{cor}
Suppose that $f:\Rr^2\to\Rr^2$ is a chaotic map where all periodic
orbits $\xf\in\Per(f)$ are of saddle type with eigenvalues 
$\lambda_1, \lambda_2$ that satisfy condition~\eqref{eq:StabCond}, 
i.e., we have
\[\ln(\abs{\lambda_2})\Re(\lambda_1) \neq \ln(\abs{\lambda_1})\Re(\lambda_2).\]
Then all periodic orbits $\xf\in\Per(f)$ are SPFC-stabilizable.
\end{cor}

Note that the number of constraints for stabilizability grows 
with increasing dimension of the dynamical system. In order to 
determine the absolute number of periodic orbits which are 
stabilizable for higher dimensional
systems, a more detailed knowledge about the ``average'' local stability
properties of periodic orbits is needed.

Since the system is real, complex eigenvalues of the derivative will
always come in complex conjugate pairs. According to
Remark~\ref{rem:Conjugate} above, this actually results in an
effective decrease of the number of constraints.

\subsection{A Geometric Interpretation}

\begin{figure*} 
\subfigure[]{\includegraphics[scale=\imagescaling]{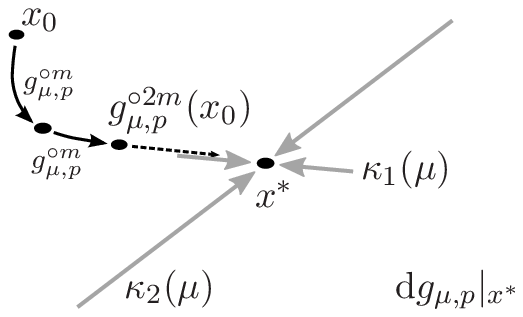}}\qquad\qquad
\subfigure[]{\includegraphics[scale=\imagescaling]{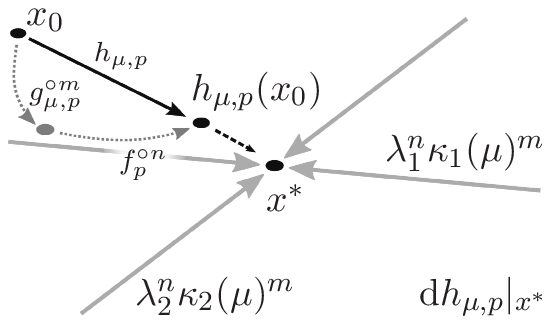}}
\caption{\label{fig:geometry}Why stalling chaos control
speed up convergence. Iteration of $g_{\mu,p}$ takes a trajectory
to the periodic orbit~$\xf$ slowly along the direction of the
originally stable manifold (Panel~(a)). Stalling control accelerates 
convergence by taking advantage of the fast convergence speed along the
stable manifold (Panel~(b)) leading to fast overall convergence speed.
The length of the
gray arrows illustrate convergence speed as they scale inversely with
the corresponding value of the eigenvalue.}
\end{figure*}

The local stability considerations also explain why Stalled
Predictive Feedback Control increase asymptotic convergence
speed\cite{Bick2012}. Consider a periodic orbit~$\xf$ of
saddle type in a two-dimensional system where contraction
along the stable direction is given by $\lambda_1\in(-1,1)$ 
and expansion along the unstable manifold by $\lambda_2<-1$.
As discussed above, these are the PFC-stabilizable periodic
orbits. Suppose that $\mu_\text{opt}>0$ is the value
of the control parameter for which the spectral radius of the
linearization of the PFC-transformed map $g_{\mu, p}$ takes its
minimum. For $\lambda_2\ll -1$ we have $\mu_\text{opt}\approx 0$
and therefore $\kappa_1(\mu_\text{opt})\approx 1$ determines
the asymptotic convergence speed of the dominating direction if
the periodic orbit is
stabilized. Therefore the trajectory will approach the periodic
orbit along the direction corresponding to~$\lambda_1$; 
cf.~Figure~\ref{fig:geometry}. The slowdown of Predictive Feedback
Control is caused by the fact that for highly unstable periodic
orbits, the trajectories converge to the originally stable
manifold along which convergence is slow in the transformed
system.

Stalling PFC exploits exactly this property. First, iteration of 
$g_{\mu, p}$ takes the trajectory closer to the stable manifold.
Second, iteration of $f_p$ leads to fast convergence along the stable
manifold while diverging from the stable manifold; 
cf.~Figure~\ref{fig:geometry}. Thus, asymptotic convergence speed
of $h_{\mu, p}$ is increased by making use of the (increasing)
stability of the stable direction. For given stalling 
parameters~$\prm,\prn$ the optimal value of the control
parameter $\mu$ is close to the zero of~$\kappa_2(\mu)$. For this
value, convergence to the stable direction is strongest, taking
full advantage of the fast convergence given by~$\lambda_{1}$ along
the stable manifold of the chaotic map~$f$. The question of how
to chose the stalling parameters~$\prm,\prn$ will be addressed 
in the following sections.

\section{Convergence Speed for Chaotic Maps}\label{sec:Example}

In the previous section we analyzed the stability properties of
the SPFC method for periodic orbits in dependence of their stability
properties. The improvements due to stalling can be 
calculated explicitly for some ``typical'' two and three-dimensional
chaotic maps.

With $\sr_{\text{min}}^{h}(\xf) = \inf_{\mu, \alpha}\sr_{\xf}(\alpha, \mu)$
denoting the rescaled stability of the linearization for the
optimal parameter values, we calculated the functions
{\allowdisplaybreaks
\begin{subequations}\label{eq:StabPerH}
\begin{align}
\label{eq:StabPerHFirst}
\overline{\rho}_h(p) &= 1-\min_{\xf\in\FPs_h(f,p)}\sr_{\text{min}}^{h}(\xf),\\
\rho_h(p) &= 1-\frac{1}{\card(\FPs_h(f,p))}\sum_{\xf\in\FPs_h(f,p)}\sr_{\text{min}}^{h}(\xf),\\
\label{eq:StabPerHLast}
\underline{\rho}_h(p) &= 1-\max_{\xf\in\FPs_h(f,p)}\sr_{\text{min}}^{h}(\xf)\end{align}
\end{subequations}
}numerically in the same fashion as \eqref{eq:StabPerG} to asses the
scaling of optimal asymptotic convergence speed of Stalled Predictive
Feedback Control for a given chaotic map
across different periods. That is, for every periodic orbit of $f$
of minimal period $p$ we calculated the spectral radius at the optimal
parameter values and then took the minimum, maximum, and mean of these
values. In particular, $1-\underline{\rho}_h$ is the upper limit and
$1-\overline{\rho}_h$ is the lower limit for the best asymptotic
convergence speed of all SPFC-stabilizable periodic orbits of a
given period~$p$ rescaled to one evaluation of $f_p$.

The increase of the number of stabilizable orbits for PFC and SPFC
can be quantified by looking at the fractions of stabilizable 
periodic orbits that are given by
\begin{equation}\label{eq:Stabilizability}
\nu_h(p) = \frac{\card(\FPs_h(f,p))}{\card(\FP(f,p))}\text{ and }
\nu_g(p) =\frac{\card(\FPs_g(f,p))}{\card(\FP(f,p))},
\end{equation}
respectively.

\subsection{Stabilizability for chaotic maps}

\begin{figure*} 
\subfigure[]{\includegraphics[scale=\imagescaling]{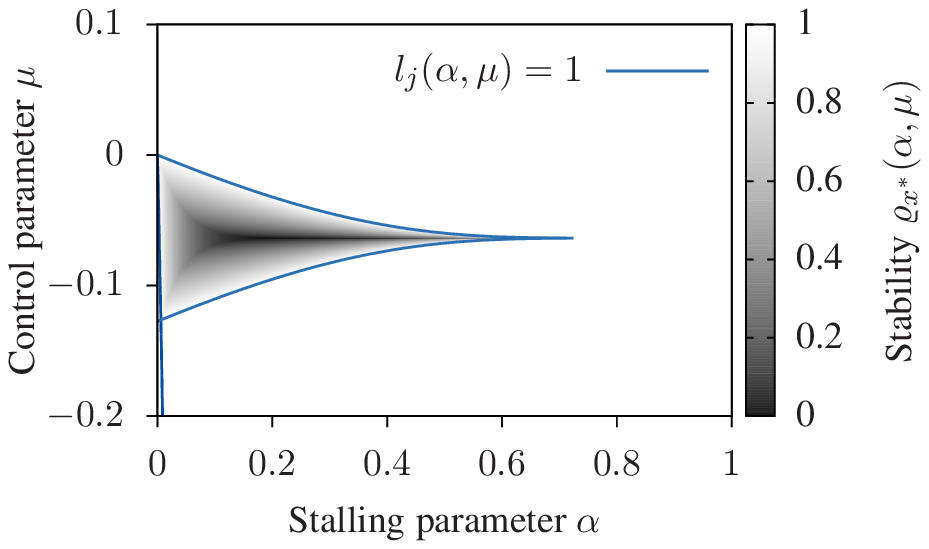}}\qquad\qquad
\subfigure[]{\includegraphics[scale=\imagescaling]{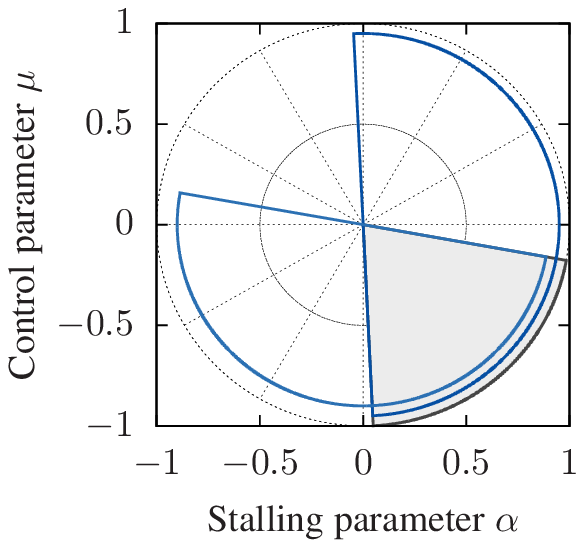}}
\caption{\label{fig:2DSingle}\CO{}Stability analysis for a periodic orbit of
period $p=5$ of the map \eqref{eq:Example2D} with local stability given
by $\lambda=(\lambda_1, \lambda_2) = (1.46\cdot10^{-9}, 16.698)$ yields
a region in parameter space in which it is stable. Panel~(a) shows the
stability function \eqref{eq:StabilityFunction} and the lines defined by
$l_j(\alpha, \mu) = 1$ with $l_j$ as given by~\eqref{eq:SingStab}. The
domain of stability $C_\lambda$ around $(\alpha, \mu)=0$ is depicted in
Panel~(b). Note that this periodic orbit cannot be stabilized using the
PFC method.}
\end{figure*}

Consider the two-dimensional neuromodule~\eqref{eq:Example2D}
discussed above and let~$\xf$ be some periodic orbit. The
stability function describes local stability at~$\xf$;
cf.~Figure~\ref{fig:2DSingle}. The region of stability in 
$(\alpha, \mu)$-parameter space is bounded by the lines 
$l_j(\alpha, \mu) = 1$ where $j=1,2$. The intersection of the half
planes defined
by the lines $\eqref{eq:Halfplane}$ gives the sector $C_{\lambda}$
that describes stability around $(\alpha, \mu)=0$ where $\lambda=(\lambda_1,
\lambda_2$) are the eigenvalues of $\drv{f_p}{\xf}$; 
cf.~Section~\ref{sect:StallCtrlProp}. Note that for 
fixed~$\alpha$, the range of $\mu$ which
yields stability becomes smaller for larger~$\alpha$.

\begin{figure}[b] 
\includegraphics[scale=\imagescaling]{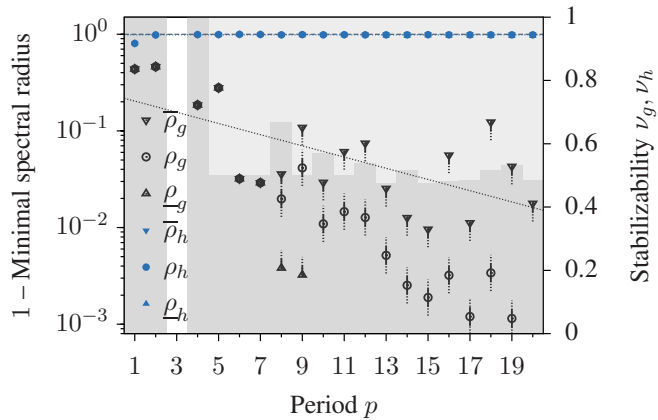}
\caption{\label{fig:ScalingStab2D}\CO{}Stalling PFC increases optimal
asymptotic convergence speed for the 2D-Neuromodule~\eqref{eq:Example2D}.
SPFC yields period-independent asymptotic convergence speed. The shading
indicates that more periodic orbits can be stabilized. The fraction of 
stabilizable orbits is shaded in gray; dark indicates stabilizability for
both with and
without stalling, light indicates stabilizability for SPFC only.}
\end{figure}

To compare the scaling of the spectral radius across periods, we plotted
the functions~\eqref{eq:StabPerG} and~\eqref{eq:StabPerH} in
Figure~\ref{fig:ScalingStab2D}. The original PFC method exhibits
asymptotic convergence speeds that approach one exponentially
for increasing period. A fit of  $\underline{\rho}$, corresponding
to the best
asymptotic convergence speed, by a function $\phi(x) = a\exp(-bx)$
yields a slope of $b=0.1334$. By contrast,
stalling the control significantly improves this scaling. We obtain
values close to zero for all periods $p\in\{1, \dotsc, 20\}$ and hence
period-independent asymptotic convergence speed in terms of evaluations
of $f_p$. A~fit with an exponential function of $\underline{\rho}_h(p)$,
i.e., the worst convergence speed, yields an exponent of 
$b=3.8112\cdot 10^{-8}$.

Qualitatively similar results are obtained for other two-dimensional
chaotic maps\cite{Bick2012} such as the H\'enon map~\cite{Henon1976}
and the Ikeda map~\cite{Ikeda1980} (not shown).

\begin{figure*} 
\subfigure[]{\includegraphics[scale=\imagescaling]{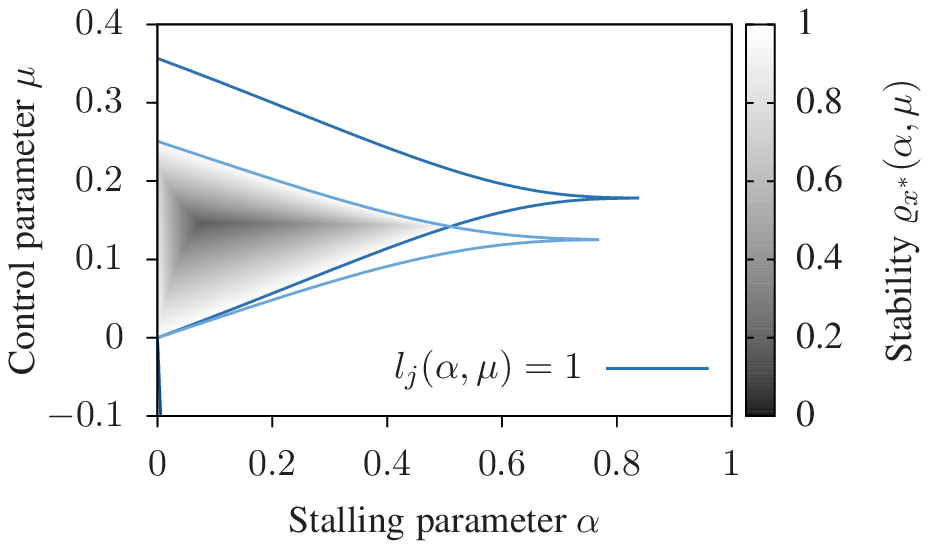}}\qquad\qquad
\subfigure[]{\includegraphics[scale=\imagescaling]{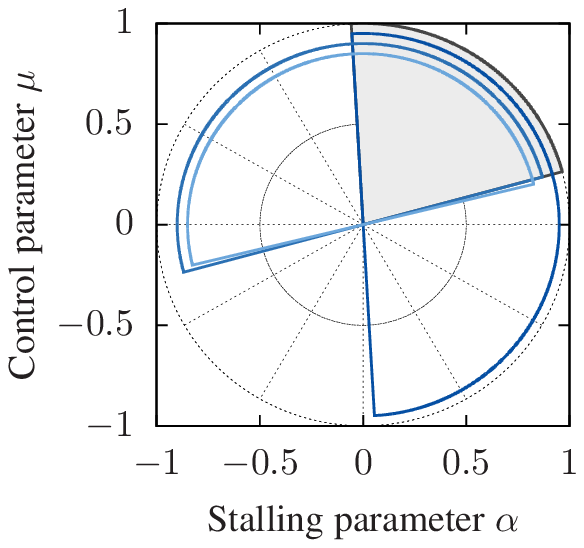}}
\caption{\label{fig:3DSingle}\CO{}Stability properties for a fixed point of
period $p=6$ of the three-dimensional H\'enon map~\eqref{eq:3DHenon}
with local stability given by $\lambda = (\lambda_1, \lambda_2, \lambda_3) =
(3.1125\cdot10^{-8}, -4.6072, -6.9734)$ show a region where stabilization
is successful. The stability function~\eqref{eq:StabilityFunction} is depicted
in Panel~(a) and the domain of stability $C_\lambda$ in Panel~(b), 
cf.~Figure~\ref{fig:2DSingle}.}
\end{figure*}

As an example of a three-dimensional system, we analyzed a
three-dimensional extension of the H\'enon map \cite{Baier1990}
given by 
\begin{equation}\label{eq:3DHenon}
f(x_1, x_2, x_3)=\left(a-x_2^2-bx_3, x_1, x_2\right)
\end{equation}
with parameters $a=1.76, b=0.1$. Stability properties of a
periodic orbit of period $p=6$ are depicted in 
Figure~\ref{fig:3DSingle}.

Due to additional constraints on stabilizability, the situation
is different compared to the two-dimensional example above. In our
example, the periodic orbits have a two-dimensional unstable manifold.
If both eigenvalues corresponding to that manifold are real, the
regime of stability depends on their sign and distance. If they have
opposite signs, the periodic orbit cannot be stabilized, neither with
nor without stalling. In case both eigenvalues have the same sign, the
situation is depicted in Figure~\ref{fig:3DSingle}; there is a maximal
value for~$\alpha$ beyond which stabilization fails. For a pair
of complex conjugate eigenvalues, the stability properties
depend on the quotient of the real and imaginary part; 
cf.~Figure~\ref{fig:PFCvSPFC}. In particular, if the imaginary
part is large, optimal asymptotic convergence speed is achieved
for the PFC method, i.e., for a choice of $\prn=0$.

\begin{figure} 
\includegraphics[scale=\imagescaling]{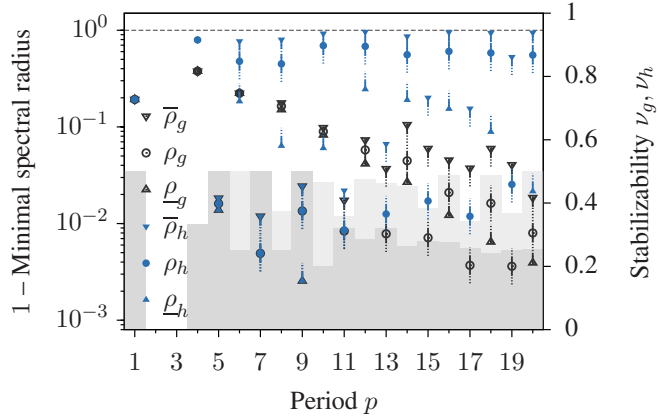}
\caption{\label{fig:ScalingStab3D}\CO{}Stalling Predictive Feedback
Control yields period-independent scaling for periodic orbits of
even period for the three-dimensional H\'enon 
generalization~\eqref{eq:3DHenon}. Effectivity of stalling
for odd periods increases with increasing period. 
The number of stabilizable periodic orbits~\eqref{eq:Stabilizability}
roughly doubles for higher periods as indicated by the shading, 
cf.~Figure~\ref{fig:ScalingStab2D}.}
\end{figure}

When looking at the scaling of optimal asymptotic convergence
speed across periods we have to distinguish between even and 
odd periods (Figure~\ref{fig:ScalingStab3D}). For even periods,
we obtain a
period-invariant scaling of both the mean and the best optimal
asymptotic convergence speed similar to the two-dimensional
system. While the upper bound on convergence speed will also
increase to one due to
the existence of periodic orbits with complex conjugate pairs,
it will typically stay above the best convergence speed for
the original PFC method. For
odd periods, the number of periodic orbits with complex
conjugate pairs of eigenvalues corresponding to the unstable
directions is large. Therefore, we see the same performance as
for the PFC method. Interestingly, for larger odd
periods $p>10$ stalling becomes more effective at increasing
optimal asymptotic convergence speed, boosting the best speed
close to one.

A similar scaling behavior is present in other three-dimensional
examples; period-independent scaling for even periods $p$
is observed for a three-dimensional
neuromodule~\cite{Pasemann2002} (not shown).

\subsection{Convergence speed in applications}\label{sec:Simulation}

The scaling of the spectral radius indicates only the best
possible asymptotic convergence speed for Stalled Predictive
Feedback Control, i.e., the speed for the linearized dynamics.
We ran simulations to compare the convergence speed for the
full nonlinear system with the theoretical results for the
linearized dynamics. In order to approximate a real-world
implementation where control is turned on at a ``arbitrary point
in time'' initial conditions were distributed randomly on the
attractor according to the chaotic dynamics.

To evaluate convergence speed of Stalled Predictive Feedback
Control, we compared the speed of $g_{\mu, p}=h^{0}_{\mu, p}$
with~$h^{\alpha}_{\mu, p}$ for both $\alpha = 3^{-1}$ and
$\alpha=(p+1)^{-1}$.
In terms of the parameters $\prm,\prn$, a value of
$\alpha = 3^{-1}$ corresponds to $\prm=2$, $\prn=1$ and 
$\alpha=(p+1)^{-1}$ to $\prm=p$, $\prn=1$.
In our implementation,
convergence time is the time~$T$ for the dynamics to satisfy
\begin{equation}\label{eq:StalledConvCrit}
\norm{x_T-\psi(x_T)}\leq\theta\ind{conv},
\end{equation}
where~$\psi$ is one of the functions above.
Convergence was only achieved if the criterion was fulfilled before
a timeout of $T\ind{timeout}=3000$ iterations. The
convergence times were rescaled to evaluations of~$f_p$ to make them
comparable. To calculate the best theoretical convergence time, we
calculated the smallest spectral radius 
\[\underline{\rho}^{\alpha}(p) = \min_{\xf\in\FPs_h(f,p)}
\inf_{\mu}\sr\big(\drv{h^{\alpha}_{\mu, p}}{\xf}\big)\]
for all periodic orbits of a given period $p$ with variable~$\mu$
while keeping the stalling parameter $\alpha(m,n)$ fixed. 
By assuming 
$\norm{\xf-x_\tau}=\norm{\xf-x_0}\big(\underline{\rho}^{\alpha}(p)\big)^{\tau}$
for the linear system we have that for an initial separation 
of $\norm{\xf-x_0}=d\ind{ini}$ the convergence 
criterion~\eqref{eq:StalledConvCrit} is satisfied for 
\begin{equation}\label{eq:TheorConvSpeed}
\tau^{\alpha}(p)=\left(\ln\!\left(\frac{\theta\ind{conv}}{d\ind{ini}}\right)-
\ln\!\left(1-\underline{\rho}^{\alpha}(p)\right)\right)\ln\!\left(\underline{\rho}^{\alpha}(p)\right)^{-1}
\end{equation}
Thus, $\tau^{\alpha}(p)$ is the convergence time of the linearized
system for an initial condition~$x_0$ with (period-independent) 
initial separation $d\ind{ini}$. For the simulations
presented here, we chose $\theta\ind{conv}=10^{-13}$ and 
$d\ind{ini}=0.1$.

\begin{figure*} 
\subfigure[]{\includegraphics[scale=\imagescaling]{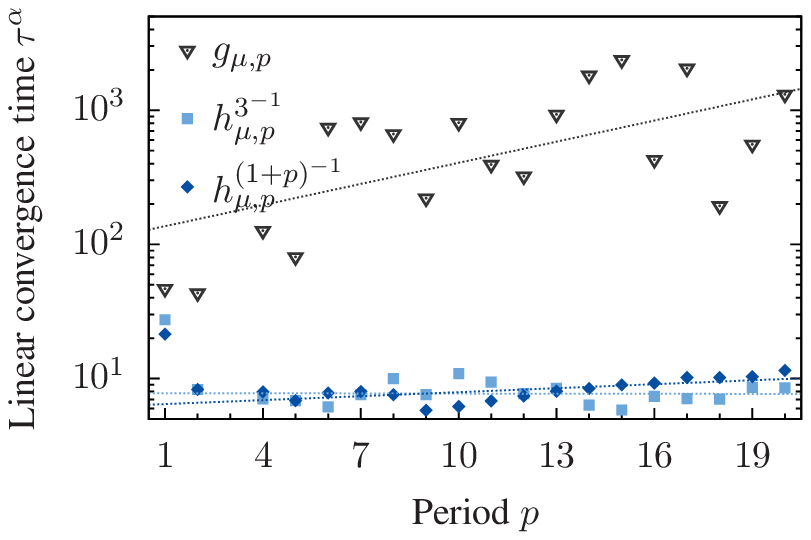}}\qquad\qquad
\subfigure[]{\includegraphics[scale=\imagescaling]{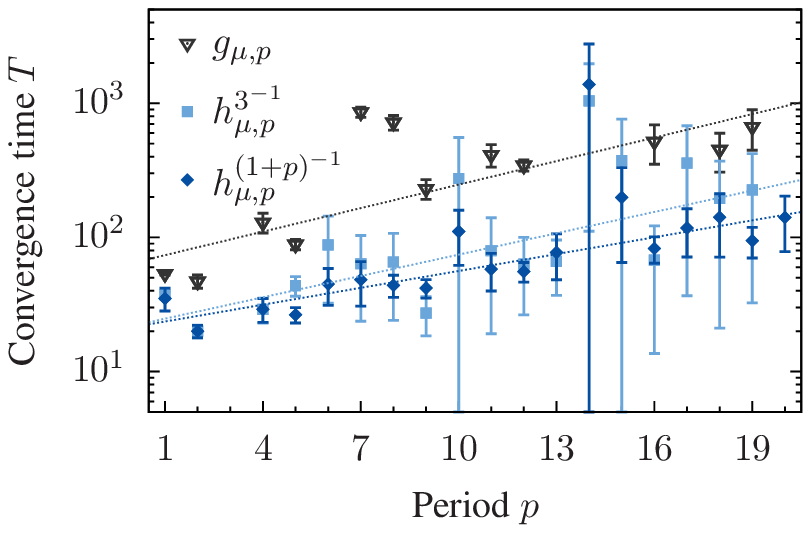}}
\caption{\label{fig:NumSpeed}\CO{}
Although the best convergence times obtained from numerical simulations
as shown in Panel~(b) cannot match the theoretical values of the linearized
system given by \eqref{eq:TheorConvSpeed}, shown in Panel~(a), stalling 
PFC increases both the overall convergence
times as well as the scaling across periods. Numerical simulations for
the two-dimensional neuromodule~\eqref{eq:Example2D} were performed
with initial conditions distributed randomly on the chaotic attractor.
Dashed lines represent an approximate exponential fit to
indicating the overall scaling behavior.}
\end{figure*}

The results are shown in Figure~\ref{fig:NumSpeed}. The errorbars
depict mean and standard deviation for all~$500$ runs with initial
conditions given by transient iteration of random length on the
attractor. The value of the
control parameter~$\mu$ in the numerical simulations was chosen for
each period to be the optimal value that yielded at least a fraction
of~$0.95$ of convergent initial conditions. In other words, $\mu$
was chosen to yield the optimal
speed with at least~$95\%$ reliability.

As predicted by the calculation of the spectral radius, stalling
PFC leads to an increase in convergence speed across all periods. 
A scaling of convergence times (scaling is indicated by dashed lines)
which is almost period-independent as observed in the theoretical
calculations cannot be achieved in our simulations. This is due
to several factors. First, in contrast to the linearized dynamics,
the numerical simulations take the full nonlinear system into 
account. This includes the influence of the transient dynamics 
and the increasing complexity of the phase space (the number
of fixed points increases with increasing period) on convergence
times. Second, in the theoretical calculations we consider only the
fixed point for which convergence is fastest. However, even in our
simulations, stalling improves both absolute convergence times
was well as their scaling across periods compared to classical
PFC. Furthermore, it increases
the number of periods that can be stabilized. For some periods,
only Stalled Predictive Feedback Control yields convergence within
a reasonable time. The scaling of the convergence speeds is 
independent of whether the stalling parameter is fixed or scales
with~$p$. However, a  period-dependent stalling parameter will
generally reduce the standard deviation of the different convergence
times.

\subsection{Relation to earlier results}\label{sec:Before}

\begin{figure*} 
\subfigure[]{\includegraphics[scale=\imagescaling]{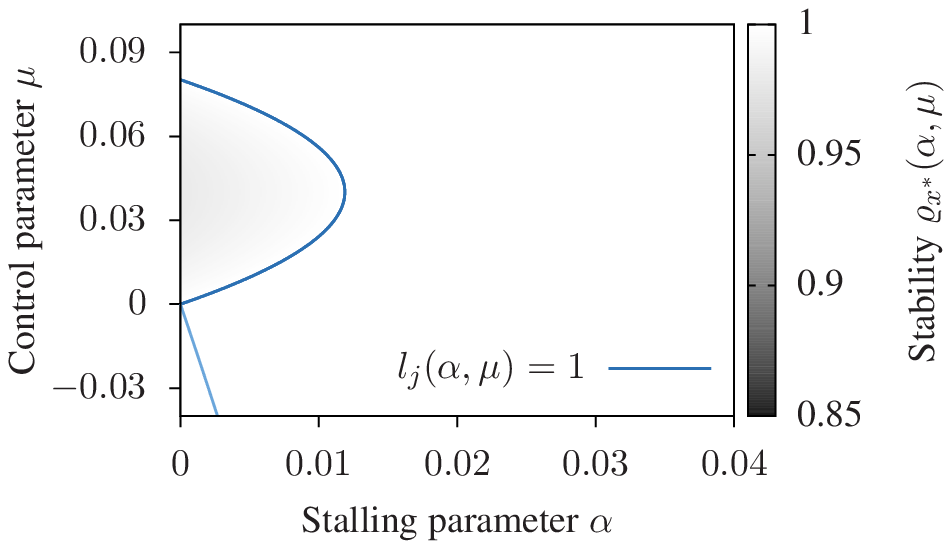}}\qquad\qquad
\subfigure[]{\includegraphics[scale=\imagescaling]{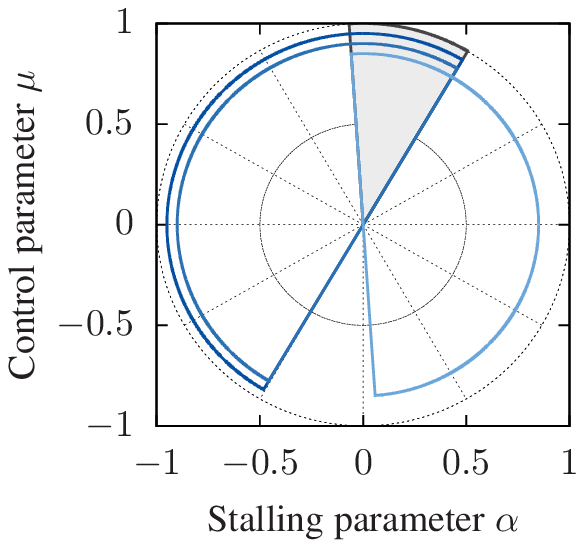}}
\caption{\label{fig:PFCvSPFC}\CO{}For a period $p=5$ orbit of the
three-dimensional H\'enon map~\eqref{eq:3DHenon} with local stability
properties $\lambda = (\lambda_1, \lambda_2, \lambda_3) = (0.0933+4.6673i,
0.0933-4.6673i, 0)$ with unstable directions given by a pair of complex
conjugated eigenvalues, only few choices of the stalling parameter allow
for stabilization (in particular, $m\gg 1$), cf.~Figures~\ref{fig:2DSingle}
and~\ref{fig:3DSingle}. Optimal performance is achieved for the PFC method,
i.e., with $n=0$.}
\end{figure*}

Stalled Predictive Feedback Control as defined in Definition~\ref{def:SPFC}
is a proper extension of the PFC method. In fact, the iteration of
$h_{\mu, 1}^{(1,1)}$ has been considered before in the context of
Predictive Feedback Control when trying to overcome the odd number
limitation~\cite{Schuster1997, Morgul2006} as well as in the context
of an experimental setup where measurements are 
time-delayed~\cite{Claussen2004}. These studies were only concerned with
whether or not fixed points can be stabilized, completely ignoring
the aspect of convergence speed. Although for systems of dimension
$\maxdim<3$ stalling control increases the number of fixed points that can be
stabilized; even for $\maxdim=3$ there are points that can be stabilized using
PFC but not using SPFC when the stalling parameter~$\alpha$ is as large
as in \cite{Schuster1997, Morgul2006, Claussen2004}
(Figure~\ref{fig:PFCvSPFC}). Hence, the introduction of an arbitrary
stalling parameter
is the key to both maximizing the number of fixed points subject to
stabilization through PFC as well as minimizing the convergence
speed.

The idea of periodically turning control on and off has been mentioned
before in the literature on control theory; both ``act-and-wait''
control~\cite{Insperger2007} and ``intermittent'' 
control~\cite{Gawthrop2010} are stated for linear control
problems in discrete and continuous time. At the same time, for linear
control problems with many control parameters, ``pole placement''
techniques~\cite{Sontag1998} are used to control the eigenvalues of
the linearization.
By contrast, SPFC aims at stabilizing many unstable periodic orbits
of a given nonlinear
system maintaining the simplicity of the simple one-parameter feedback
control scheme. The situation where control is turned on at
an arbitrary point in time as described above is of particular
interest; here, the system is likely to be far from the linear
regime. As shown above, stalling PFC improves performance even in
this situation.

Stalling Predictive Feedback Control is also related to a recent
application of chaos control~\cite{Steingrube2010}. Because of
implementation restraints, Steingrube~{et.~al.~effectively}
iterated \mbox{$f\circ g_{\mu, p}$}. In some sense, this is similar
to iterating $h_{\mu, p}^{(p, 1)}$, but the stability analysis is
not straightforward
since one has to keep track of the (changing) point on the periodic
orbit to be stabilized when iterating \mbox{$f\circ g_{\mu, p}$}. 
Moreover, both this control and SPFC are
related to an effort by Polyak~\cite{Polyak2005} to introduce a
generalized PFC method, which is capable of stabilizing periodic
orbits with an arbitrary small perturbation. This method, however,
is limited in applicability, because the control perturbation depends on
predictions of the state of the system many time steps in the
future.

\section{Adaptive Control}\label{sec:Adaptation}

In the previous sections, we showed that for an optimal choice of
parameters the asymptotic convergence speed of Predictive Feedback
Control can be significantly increased when stalling control.
This speedup is not only of theoretical nature, but also persists
in an implementation with random initial conditions. Bot how does
one find the set of optimal parameter values for a given chaotic 
map~$f$? If no a~priori estimates are available, adaptation methods 
provide a way to tune the control parameters online for optimal 
convergence speed.

Here, we consider the case where the stalling parameter~$\alpha$
(corresponding to some choice of $m,n$) is fixed and $\mu\geq 0$
is subject to adaptation. We explore different
adaptation mechanisms and propose a hybrid gradient adaptation
approach that leads to fast and highly reliable adaptation across
different periods for initial conditions distributed randomly on
the chaotic attractor.

\subsection{Simple and gradient adaptation}

First, recall a simple adaptation scheme\cite{Steingrube2010}. We 
assume that the period~$p$ is fixed within this subsection. A 
suitable objective function for 
finding a periodic point of period~$p$ is given by
\[G_1(x, p) = \norm{f_p(x)-x}^2\]
for some vector norm $\enorm$ on $\Rn$. For $\mu = 0$ the map~$h_{0, p}$
as defined in~\eqref{eq:DelayedStab}.
reduces to some iterate of~$f$ and
adaptation should lead to sequences $x_k\to\xf$ and $\mu_k\to\mu^*$
with $\xf\in\FP(f_p)$ and $\sr_h(\alpha, \mu^*) < 1$. The
objective function above suggests a simple adaptation rule (SiA) with
\begin{equation}\label{eq:Adapt}
\Delta\mu_k = \nu(p) G_1(x_k, p)
\end{equation}
where $\nu(p)$ is the (possibly period-dependent) adaptation parameter
and dynamics of $\mu$ given by
\begin{equation}\label{eq:MuChange}
\mu_0=0,\quad \mu_{k+1} = \mu_k+\Delta\mu_k.
\end{equation}
This adaptation rule increases the control parameter $\mu$ monotonically.
Suppose that~$\xf$ is a fixed point of $f$, i.e., \mbox{$f_p(\xf)=\xf$}.
If we have a converging sequence $x_k\to \xf$ as $k\to\infty$ then
the sequence $\Delta\mu_k$ tends to zero. In other words, adaptation
stops in the vicinity of a fixed point~$\xf$ of~$f_p$.

For this adaptation mechanism, the quantity $\Delta\mu_k$ is
extremely easy to calculate and yields decent results in
applications~\cite{Steingrube2010}. Adaptation, however, strongly depends
on the choice of the adaptation parameter~$\nu(p)$. If~$\nu(p)$ is too
small, it will take a long time
to reach a regime in which convergence takes place. On the other hand,
if~$\nu(p)$ is too large and the interval~$M$ of possible values 
of~$\mu$ in
which convergence takes place is rather narrow, it is possible that
$\mu_k > \sup M$ for some $k$, even if $\mu_{l}\in M$ for some values
$l<k$. Hence, it is possible for the control parameter to ``jump
out of'' the range of stability. Also, note
that by construction, this simple adaptation will not optimize for
asymptotic convergence speed. For small $\nu(p)$, adaptation will 
stop close to the boundary of the convergent regime, leading to slow
asymptotic convergence speed; cf.~Figure~\ref{fig:MuDyn}~(a).

Adaptation may be improved, if the objective function takes local
stability into account. For some matrix norm $\enorm$, such an objective
function is given by
\[G_2(x, \mu, p)=\norm{\drv{h_{\mu, p}}{x}}\]
Since any matrix norm is an upper bound for the spectral radius $\vrho(A)$
of a matrix $A$, that is $\vrho(A)\leq\norm{A}$,
minimizing the norm potentially leads to increased convergence 
speed\cite{Bick2010b}. At the same time,
for a generic point on the attractor, this objective function is highly
nonconvex with steep slopes (Figure~\ref{fig:GlobalOpt}) making
straightforward minimization through, for example, gradient 
descent~\cite{Fradkov1998} difficult.

\begin{figure} 
\includegraphics[scale=\imagescaling]{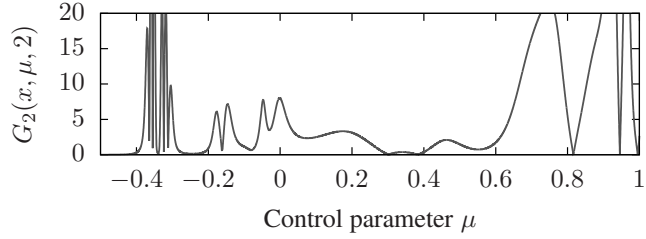}
\caption{\label{fig:GlobalOpt}The objective function $G_2(x, \mu, p)$ is
nonconvex for a generic point $x$ on the attractor leading to a
difficult optimization problem.}
\end{figure}

We therefore propose an adaptation rule that combines aspects of
simple adaptation as reviewed above and the objective function~$G_2$.
Let~$\partial_\mu$ denote the derivative with respect to~$\mu$
and define $\Theta(x) = \tanh\left((pG_1(x, p))^{-1}\right)$. Consider
the modified gradient adaptation rule (GrA) given by~\eqref{eq:MuChange}
with
\begin{multline}\label{eq:AdaptNew}
\Delta\mu_k = \lambda(p)\left(G_1(x, p)\right. \\ -\left. p\tanh\left(\Theta(x_k) \partial_\mu G_2(h_{\mu, p}(x_k), \mu, p)\right)\right).
\end{multline}
This adaptation rule has the following properties. Far
away from a period~$p$ orbit $\xf\in\FP(f_p)$, i.e. for 
$G_1(x, p)\gg 0$, we have 
$\Theta(x)\approx 0$. Therefore, adaptation is dominated by the 
first term and leads to adaptation as given by the simple adaptation 
rule~\eqref{eq:Adapt} to increase~$\mu$ to reach a regime of
convergence. On the other hand, in the vicinity of a fixed point
we have $\Theta(x)\approx 1$ and $G_1(x, p)\approx 0$. Hence, adaptation
occurs by bounded gradient descent and the dynamics of the control
parameter~$\mu$ are perpendicular to the level sets of the objective
function~$G_2$ towards a (local) minimum. The bound induced by 
the~$\tanh$ prevents large fluctuations of the objective function~$G_2$
from leading to a too large change of the control 
parameter~$\mu$.

\begin{figure*}[h]
\subfigure[]{\raisebox{3pt}{\includegraphics[scale=\imagescaling]{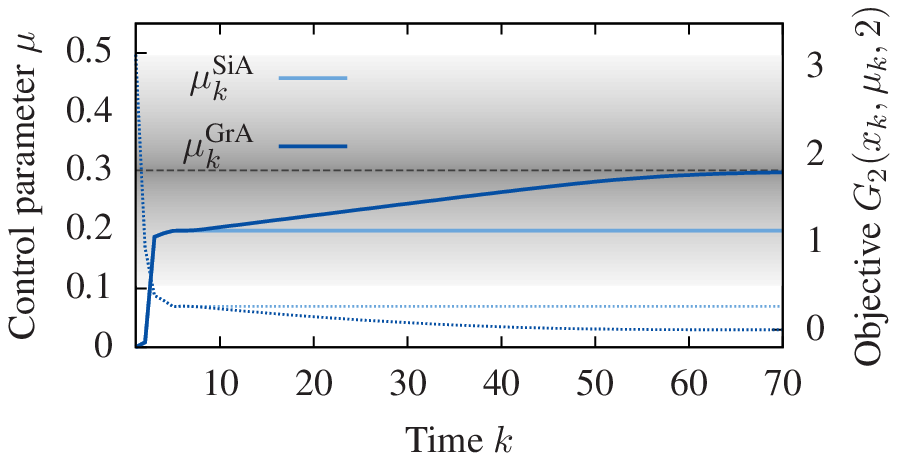}}}\qquad
\subfigure[]{\includegraphics[scale=\imagescaling]{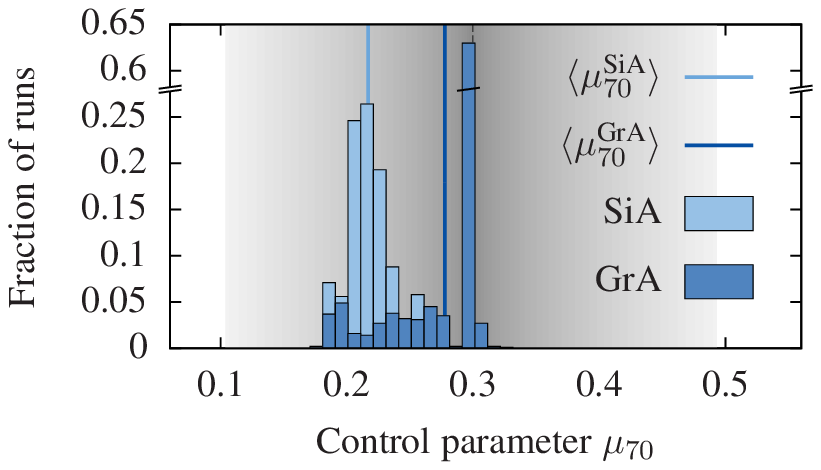}}
\caption{\label{fig:MuDyn}\CO{}In contrast to simple adaptation, gradient
adaptation tunes the control parameter to the value where optimal convergence
speed is achieved. The dynamics for a single run are shown in Panel~(a)
and the dotted lines depict the value of the objective function~$G_2$.
Statistics for 1000 initial conditions on the attractor after a transient of 
random length are shown in Panel~(b). The shading indicates
values of the stability function smaller than one and the dashed line its
minimum (optimal asymptotic convergence speed).
The target period was $p=2$ for the two-dimensional
map~\eqref{eq:Example2D} with adaptation parameter $\nu=10^{-3}$
and $\prn=1$, $\prm=2$. Here, $\langle\ \cdot\ \rangle$ denotes the
population mean.
}
\end{figure*}

\begin{figure*}[h]
\subfigure[]{\includegraphics[scale=\imagescaling]{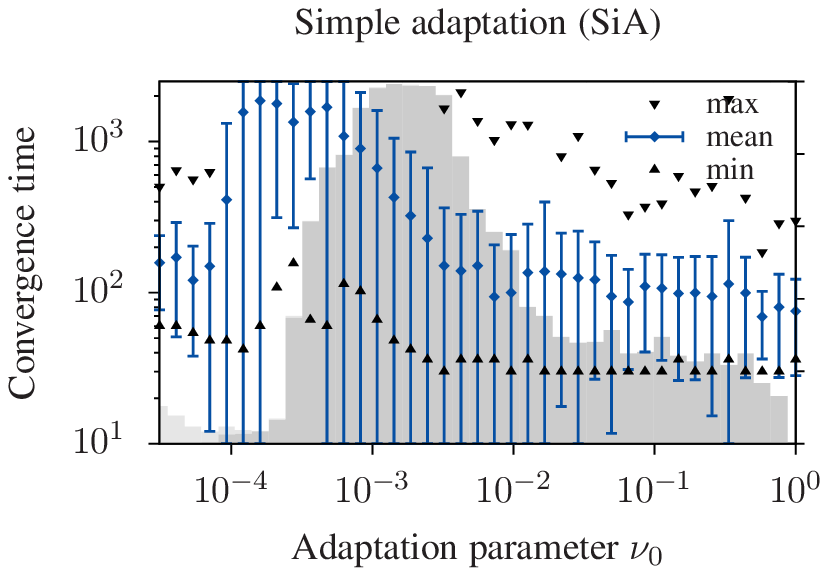}}\qquad
\subfigure[]{\includegraphics[scale=\imagescaling]{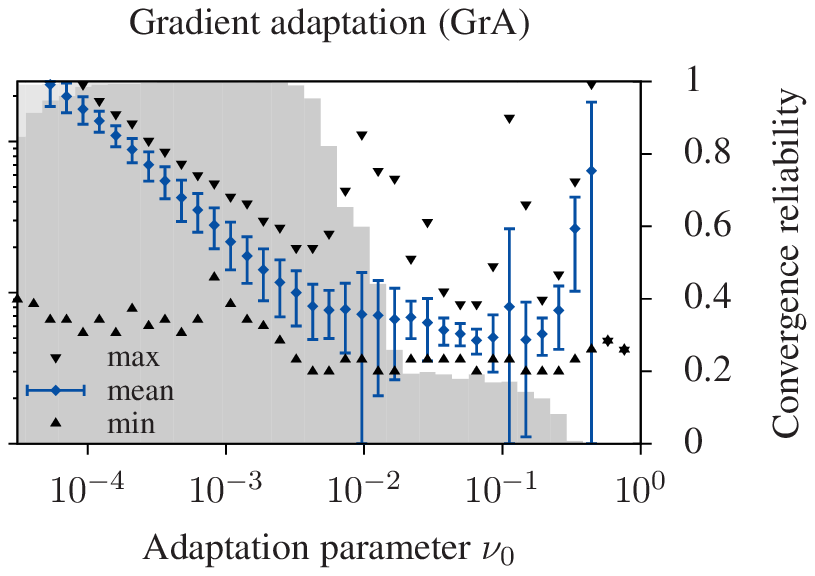}}
\caption{\label{fig:SinglePer}\CO{}Gradient adaptation (Panel~(b)) decreases the
overall convergence times and the variaton thereof compared to simple
adaptation (Panel~(a)) for target period $p=5$. Furthermore, the range
of reliable convergence, depicted by the shading in the background, is
broadened. The fraction of convergent runs to a periodic orbit of the
correct period is shaded in dark gray (reliable convergence) and to an
incorrect period in light gray.}
\end{figure*}

\begin{figure*}[h]
\subfigure[]{\includegraphics[scale=\imagescaling]{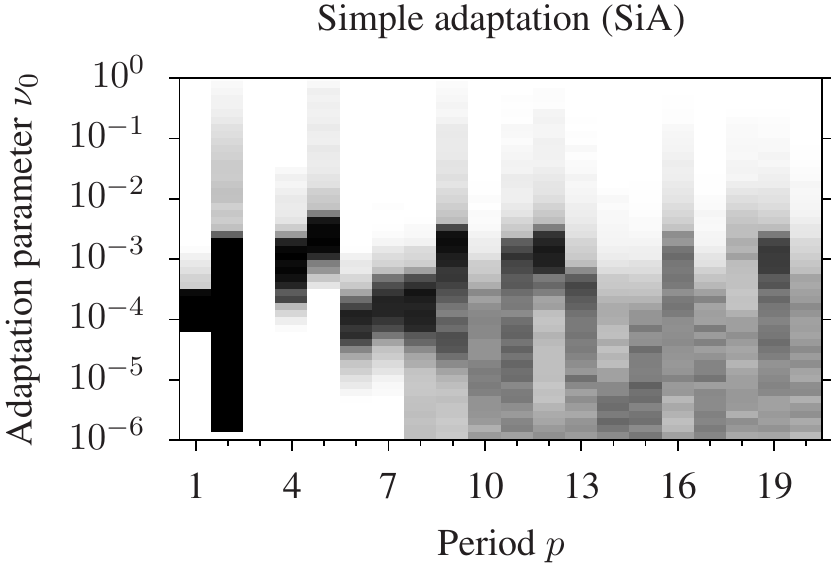}}\qquad
\subfigure[]{\includegraphics[scale=\imagescaling]{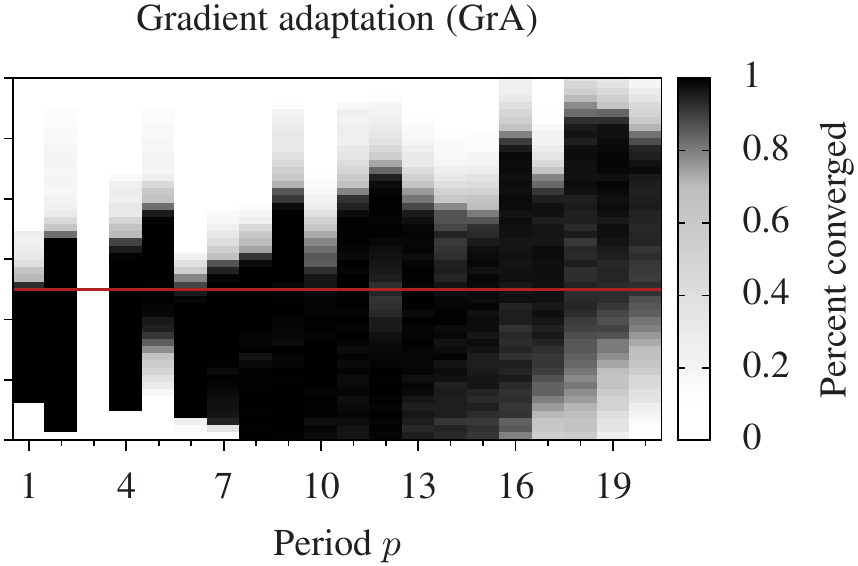}}
\caption{\label{fig:Reliability}Gradient adaptation (Panel (b)) increases
the overall reliability of convergence across periods compared to
simple adaptation (Panel (a)). Reliability, i.e., the percentage of 
convergent runs to periodic orbits of the target period $p$, is depicted 
by the color
of the shading for the adaptation parameter $\nu_0$ and hence more
dark areas correspond to higher overall reliability. For gradient 
adaptation there exist parameter values which yield reliable 
convergence across all periods ($\nu_0=10^{-3.5}$ is depicted by a red line).}
\end{figure*}

The adaptation parameter $\nu(p)$ again determines the size of
the adaptation steps. In contrast to the simple adaptation method,
the modified gradient adaptation adapts bidirectionally in order
to minimize both objective functions $G_1$ and $G_2$ as depicted in
Figure~\ref{fig:MuDyn}(a). Clearly, the control parameter is adapted
to the regime of stability of a periodic orbit by the modified
gradient adaptation and $\Delta\mu_k\to 0$ as optimal asymptotic
convergence speed is achieved. Statistics for a large number of
initial conditions show that the population mean $\langle\mu_k\rangle$
for many runs is already close to the optimal value after only 70
iterations; cf.~Figure~\ref{fig:MuDyn}(b).

\subsection{Convergence reliability}

To assess the performance of the adaptive Stalled Predictive
Feedback Chaos Control algorithm in a real-world application 
we performed large scale numerical simulations for the 
two-dimensional neuromodule~\eqref{eq:Example2D}. Periodic 
orbits were stabilized using SPFC~\eqref{eq:DelayedStab} with 
the incorporation of the adaptation mechanisms given 
by~\eqref{eq:Adapt} 
and~\eqref{eq:MuChange}. The scaling of the adaptation parameter 
was given
by $\nu(p)=\frac{\nu_0}{p}$ and for every~$\nu_0$ we iterated 
for 500 initial conditions distributed randomly on the chaotic 
attractor by iterating for a transient of random length. To 
determine reliability, i.e., the fraction of runs where
the trajectory converged to a periodic orbit of the desired period,
we checked the period of the limiting periodic orbit (if any) to a
threshold of $\theta=10^{-6}$.

As discussed above, the adaptation parameter $\nu_0$ influences both
speed and reliability. The results for period $p=5$ are plotted in 
Figure~\ref{fig:SinglePer}. We find that Gradient Adaptation not 
only decreases
the total number of time steps needed to fulfill the convergence
criterion but it also decreases the overall variation across runs
(the standard deviation is depicted as an error bar). Of particular
interest for applications is the range where convergence is highly
reliable. In contrast to the simple adaptation scheme, for gradient
adaptation the range of adaptation parameter values leading to
highly reliable convergence is broadened. On the one hand, the 
gradient adaptation method optimizes for convergence speed, thereby 
increasing the chance
that the convergence criterion is fulfilled before the timeout. At the
same time, the bidirectional adaptation decreases the likelihood of 
the control parameter leaving the regime of convergence. Gradient 
adaptation therefore
improves both overall convergence speed while reducing its 
variation and increasing the reliability of control.

The improvement of reliability compared to the simple adaptation
scheme can be seen across all periods; cf.~Figure~\ref{fig:Reliability}.
The broad range of adaptation parameters giving highly reliable
convergence allows for the choice of an adaptation parameter
$\nu_0$ that will lead to reliable convergence across different
periods, effectively eliminating this parameter.

Similar results are obtained for numerical simulations for
other two- as well as three-dimensional chaotic maps (not shown). These
include the H\'enon map\cite{Bick2012} and a three-dimensional 
neuromodule\cite{Pasemann2002}.
Convergence speed of~$\mu_k$ to the optimal parameter value
can be further increased by using higher order methods, such
as Newton's method (not shown). The use of higher order methods (also with
respect to comparing simple and gradient adaptation) comes with
a higher absolute computational cost. For any 
implementation the improvement always needs to be related to the
effective improvement.

\section{Discussion}

In this article, we studied the effect stalling has on Predictive
Feedback Control. By stalling control, the inherent speed limit of
standard Predictive Feedback Control may be overcome. We highlighted
that only by taking all possible stalling parameters into account,
the maximum number of periodic orbits can be stabilized. The 
conditions on stabilizability that we derived show that stabilizability
is reduced to the conditions imposed by the eigenvalues corresponding
to the unstable directions. Stalling is very easy to implement
and, in addition to increasing convergence speed, the resulting
chaos control method is capable of stabilizing more
periodic orbits. Using numerical simulations we showed that in
applications where chaos control is turned on at a random point
in time, convergence speed is greatly improved across all periods.
Although our method was stated in terms of discrete time dynamical
systems, it also applies to continuous time dynamics if discretized
for example through a Poincar\'e map.

As examples we studied ``typical'' low-dimensional chaotic systems.
In higher dimensions, for example when studying chaotic collective
effects in networks, we expect our method to behave qualitatively
similar as in the three-dimensional case, although an increase in
dimension of the unstable manifold of periodic orbits places
additional constraints on stabilizability. A priori estimates of
the local stability properties of the periodic orbits embedded in
the attractor yield an estimate of how many periodic orbits can be
stabilized. This limitation could be overcome by tuning the
eigenvalue corresponding to some eigenvector separately. From
a mathematical point of view, a different approach would be to allow
the control parameter to take complex values, turning the problem
into one of complex dynamics in several complex variables
\cite{Bick2010b}. On the other hand, the
local stability property conditions provide design principles for
attractors to contain many unstable periodic orbits that our
Stalled Predictive Feedback Control method is capable of
stabilizing. These important questions, however, are beyond
the scope of the current article and will have to be addressed
in further research.

Conversely, the local stability properties and the narrowing
of the regime of stability for the control parameter $\mu$ while
$\alpha > 0$ is fixed can actually be exploited. Different local
stability properties of the unstable periodic orbits allow for
stabilization of a specific set of periodic orbits. Hence, through
the choice of parameters, the targeted periodic orbits can become
stable periodic orbits of the dynamics.

Adaptation mechanisms not only provide a way to tune the
adaptation parameter to a suitable value, but they also allow
for an increase in both speed and reliability. In contrast to
previously proposed adaptation~\cite{Steingrube2010, Lehnert2011}, 
the proposed hybrid algorithm also adapts for optimal convergence 
speed. A~broad range
of parameters allows for a period-independent choice of adaptation
parameter, hence giving a chaos control method with a set of
parameters for which it stabilizes many periodic points of most
periods quickly and reliably. Adaptation using the objective
function~\eqref{eq:Adapt} also prevents the system from converging
to one of the periodic orbits potentially induced by stalling control.
However, as our adaptation method merely serves as a proof of
concept, it still leaves room for improvement. In particular,
the cap of adaptation speed through the sigmoidal function is
a major source of slowdown. Moreover, adaptation
could be extended to the stalling parameter~$\alpha$.

Since stalling PFC increases the number of evaluations of~$f_p$
needed for a single iteration of $h_{p, \mu}$, it would be 
desirable to extend the theory to a ``fractional stalling parameter,'' i.e.,
to allow for stalling by composing with $\ite{f}{q}$ where $q<p$. With such
stalling, however, one needs to track the point of the periodic
orbit, as discussed in Section~\ref{sec:Before}, rendering the
theoretical analysis more subtle.

In conclusion, Stalled Predictive Feedback Control of Chaos
together with a suitable adaptation scheme is a step towards
a fast, reliable, easy-to-implement, and broadly applicable
chaos control method. It would be interesting to see it applied
in experimental setups in the future.

\section*{Acknowledgements}

CB would like to thank Laurent Bartholdi for
making this project possible. This work was supported by the Federal
Ministry of Education and Research (BMBF) by grant numbers 01GQ1005A and
01GQ1005B.

\bibliography{ChaosControl}

\end{document}

%% file: commands.tex
\newtheorem{thm}{Theorem}[section]
\newtheorem{cor}[thm]{Corollary}
\newtheorem{lem}[thm]{Lemma}
\newtheorem{prop}[thm]{Proposition}

\newtheorem{defn}[thm]{Definition}

\newtheorem{rem}[thm]{Remark}

\DeclareMathOperator{\id}{id}

\DeclareMathOperator{\card}{\#}
\DeclareMathOperator{\Rep}{Re}

\DeclareMathOperator{\Per}{Per}

\DeclareMathOperator{\FP}{Fix}

\DeclareMathOperator{\grad}{grad}

\newcommand{\maxdim}{N}

\newcommand{\C}{\mathbb{C}}
\newcommand{\Cs}{\C^{\times}}

\newcommand{\Z}{\mathbb{Z}}

\newcommand{\Q}{\mathbb{Q}}
\newcommand{\R}{\mathbb{R}}
\newcommand{\Rr}{\R}
\newcommand{\Rn}{\R^\maxdim}
\newcommand{\N}{\mathbb{N}}

\newcommand{\Hp}{\mathcal{H}}

\renewcommand{\S}{\mathbf{S}}

\newcommand{\So}{\S^{1}}

\newcommand{\Cc}{\mathcal{C}}

\newcommand{\ud}{\mathrm{d}}

\newcommand{\s}{\sigma}

\newcommand{\rset}[2]{\left\lbrace\, #1\,\left|\;#2\right.\right\rbrace}

\newcommand{\set}[2]{\rset{#1}{#2}}
\newcommand{\tset}[2]{\big\lbrace #1\,\big|\;#2\big\rbrace}
\newcommand{\sset}[1]{\left\lbrace #1\right\rbrace}
\newcommand{\tsset}[1]{\big\lbrace #1\big\rbrace}
\newcommand{\abs}[1]{\left| #1 \right|}

\renewcommand{\Re}{\Rep}
\newcommand{\ind}[1]{_{\mathrm{#1}}}

\newcommand{\vrho}{\varrho}
\newcommand{\sr}{\vrho}

\newcommand{\ite}[2]{#1^{\circ #2}}
\newcommand{\xf}{{x^{*}}}

\newcommand{\alts}[1]{\hat{#1}}

\newcommand{\norm}[1]{\|#1\|}
\newcommand{\enorm}{\norm{\,\cdot\,}}

\newcommand{\drv}[2]{\ud#1|_{#2}}

\newcommand{\prn}{n}
\newcommand{\prm}{m}

\newcommand{\pn}[2]{#1 \backslash \left\lbrace #2\right\rbrace}
\newcommand{\FPs}{\FP^{*}}